\documentclass[a4paper,11pt]{amsart}

\usepackage{amsfonts, latexsym, amsmath, amssymb, amsthm, amscd}
\usepackage[all]{xy}
\usepackage[english]{babel}
\usepackage[utf8]{inputenc}
\usepackage[T1]{fontenc}
\usepackage{graphicx}
\usepackage{booktabs}
\usepackage{braket}
\usepackage{mathtools}
\usepackage{array, tabularx}
\usepackage{setspace}
\usepackage{textcomp}
\usepackage{tikz}
\usepackage{multirow}
\usepackage{paralist}
\usepackage{comment}
\usepackage{url}
\usepackage{tensor}

\usepackage[hypertexnames=false,
backref=page,
    pdftex,
    pdfpagemode=UseNone,
    breaklinks=true,
    extension=pdf,
    colorlinks=true,
    linkcolor=blue,
    citecolor=blue,
    urlcolor=blue,
]{hyperref}

\usepackage[top=2.0in, bottom=2.0in, left=1.5in, right=1.5in]{geometry}

\newcommand{\referenza}{}

\newtheorem{thm/}{Theorem}
\newtheorem{prop/}[thm/]{Proposition}
\newtheorem{thm}{Theorem}[section]
\newtheorem*{thm*}{Theorem (\referenza)}
\newtheorem{cor/}[thm/]{Corollary}
\newtheorem{cor}[thm]{Corollary}
\newtheorem*{cor*}{Corollary \referenza}
\newtheorem{lem}[thm]{Lemma}
\newtheorem*{lem*}{Lemma \referenza}

\newtheorem{clm/}{Claim}
\newtheorem{prop}[thm]{Proposition}
\newtheorem*{prop*}{Proposition \referenza}

\newtheorem*{conj*}{Conjecture \referenza}
\newtheorem{rmk}[thm]{Remark}

\newtheorem{que}[thm]{Question}

\newcommand{\N}{\mathbb{N}}
\numberwithin{equation}{section}

\allowdisplaybreaks

\begin{document}

\title[On Abelian and Cyclic Group Codes]{On Abelian and Cyclic Group Codes}

\author{Angelo Marotta}
\address{{Angelo Marotta} - Dipartimento di matematica e applicazioni, Università degli studi di Milano-Bicocca, Via Roberto Cozzi 55, 20125 Milano (Italy)}
\email{a.marotta10@campus.unimib.it}

\begin{abstract}
We determine a condition on the minimum Hamming weight of some special abelian group codes and, as a consequence of this result, we establish that any such code is, up to permutational equivalence, a subspace of the direct sum of $s$ copies of the repetition code of length $t$, for some suitable positive integers $s$ and $t$. Moreover, we provide a complete characterisation of permutation automorphisms of the linear code $C=\bigoplus_{i=1}^{s}Rep_{t}(\mathbb{F}_{q})$ and we establish that such a code is an abelian group code, for every pair of integers $s,t\geq1$. Finally, in a similar fashion as for abelian group codes, we give an equivalent characterisation of cyclic group codes.
\end{abstract}
\maketitle


\section*{Introduction}
Throughout this paper all groups and fields are supposed to be finite. Let $G$ be a group of order $n$ and $\mathbb{F}_{q}$ be the finite field with $q$ elements, where $q$ is a power of an arbitrary prime number $p$, and denote by $\mathbb{F}_{q}[G]$ the group algebra over $\mathbb{F}_{q}$ associated to the group $G$. In the special case in which $G$ is a cyclic group of order $m$, we will denote it by $\mathcal{C}_{m}$. Furthermore, from now on, we set with $\mathbb{F}_{q}^{n}$ the $n$-dimensional $\mathbb{F}_{q}$-vector space and with $\mathcal{B}=\{e_{1},\ldots,e_{n}\}$ its standard basis.\\
Let $C\leq\mathbb{F}_{q}^{n}$ be a linear code. Then, following \cite{BRS09}, we say that $C$ is a {\it left G-code} (respectively, a {\it right $G$-code}; a {\it $G$-code}) if there exists a bijection $\phi\colon\mathcal{B}\to G$ such that its corresponding $\mathbb{F}_{q}$-linear extension, $\tilde{\phi}\colon\mathbb{F}_{q}^{n}\to\mathbb{F}_{q}[G]$, maps $C$ to a left ideal (respectively, a right ideal; a two-sided ideal) of $\mathbb{F}_{q}[G]$. In this case, we say that the linear code $C$ is {\it permutation equivalent} to the left (respectively; right, two-sided) ideal $\tilde\phi(C)$. In particular, a linear code $C\leq\mathbb{F}_{q}^{n}$ is said to be an {\it abelian group code} (respectively, a {\it cyclic group code}; a {\it metacyclic group code}) if $C$ is $G$-code for some abelian group (respectively, cyclic group; metacyclic group) $G$. More in general, if $\mathcal{G}$ is a class of groups then we say that $C$ is (respectively; a left $\mathcal{G}$-group code, a right $\mathcal{G}$-group code) a $\mathcal{G}$-group code if there exists some group $G\in\mathcal{G}$ such that $C$ is (respectively; a left $G$-code, a right $G$-code) a $G$-code.\\
In \cite{BRS09} Bernal, del Río and Simon gave a necessary and sufficient condition for a linear code to be a (left, right) group code in terms of its permutation automorphisms group and, by virtue of this, they established that any $G$-code is an abelian group code, whenever $G$ has an abelian decomposition. As a consequence of this, they established that any metacyclic group code is an abelian group code \cite[Corollary 3.2]{BRS09}, extending a well known result of Sabin and Lomonaco \cite{SL95} on split-metacyclic group code. By considering the characterisation of group codes provided in \cite{BRS09}, Pillado et al. \cite{PGMMM19} have found an equivalent description of abelian group codes by means of which they established that any group code of dimension at most $3$ can be realised as an abelian group code.\\
In this paper we use the ideas of Pillado et al. to prove our main results.

\begin{thm/}\label{thm1}
Let $C\leq\mathbb{F}_{q}^{n}$ be an abelian group code. Let $G$ be a regular subgroup of $S_{n}$ such that $\langle G,C_{S_{n}}(G)\rangle\leq PAut(C)$ and such that $C$ is permutation equivalent to some ideal of the group algebra $\mathbb{F}_{q}[G]$, on which the derived subgroup $G'$ acts trivially from the left (from the right). If $G'\neq\{1\}$, then $C$ is a $|G'|$-divisible code and $|G'|\mid w(C)$.
\end{thm/}

\begin{cor/}\label{thm2}
Let $C\leq\mathbb{F}_{q}^{n}$ be an abelian group code.  Let $G$ be a regular subgroup of $S_{n}$ such that $\langle G,C_{S_{n}}(G)\rangle\leq PAut(C)$ and  such that $C$ is permutation equivalent to some ideal of the group algebra $\mathbb{F}_{q}[G]$, on which the derived subgroup $G'$ acts trivially from the left (from the right). If $G'\neq\{1\}$, then there exist positive integers $s,t>1$, with $s\cdot t=n$, such that $C$ is permutation equivalent to a linear subspace of $\bigoplus_{i=1}^{s}Rep_{t}(\mathbb{F}_{q})$.
\end{cor/}

Motivated by Corollary~\ref{thm2}, it seems natural to ask whether the extremal codes, that is codes of the form $\bigoplus_{i=1}^{s}Rep_{t}(\mathbb{F}_{q})$, may be realised as abelian group codes. A full answer to this question is given in the following result, which leads us to obtain also new several examples of abelian group codes.

\begin{thm/}\label{thm3}
The linear code $C=\bigoplus_{i=1}^{s}Rep_{t}(\mathbb{F}_{q})$ is an abelian group code, for every $s,t\geq1$.
\end{thm/}

Finally, starting from the equivalent description of abelian group codes provided by Pillado et al. \cite[Corollary 1.4]{PGMMM19}, we establish, in a similar manner, an equivalent characterisation of those linear codes which admit a realisation as cyclic group codes.

\begin{thm/}\label{cor:nsccgc3}
Let $C\leq\mathbb{F}_{q}^{n}$ be a linear code. Then $C$ is a cyclic group code if and only if there exists a regular subgroup $G$ of $S_{n}$ such that $\langle G,C_{S_{n}}(G)\rangle\leq PAut(C)$, $G'$ is an Hall co-cyclic subgroup of $G$ and such that $C$ is permutation equivalent to some ideal of the group algebra $\mathbb{F}_{q}[G]$ on which the derived subgroup $G'$ acts trivially from the left (from the right).
\end{thm/}

This note is organized as follows. In Section~\ref{Section1}, following \cite{BRS09} and \cite{PGMMM19}, we introduce some necessary notation and results about group codes and abelian group codes and we prove Theorem~\ref{thm1} and Corollary~\ref{thm2}. In Section~\ref{Section2} we prove Theorem~\ref{thm3} and hence we give new several constructions of abelian group codes. Finally, in Section~\ref{Section3}, we prove Theorem~\ref{cor:nsccgc3} by using similar arguments as in \cite{PGMMM19}. In addition, we discover a wide class $\mathcal{G}$ of finite groups whose algebraic properties turn every $\mathcal{G}$-group code into a cyclic group code.

\section{Preliminaries}\label{Section1}
Given the finite field $\mathbb{F}_{q}$ and a positive integer $n\in\N$ we denote by $Rep_{n}(\mathbb{F}_{q})$ the repetition code of length $n$ over $\mathbb{F}_{q}$, where $Rep_{n}(\mathbb{F}_{q})=\{(\lambda,\ldots,\lambda)\ |\ \lambda\in\mathbb{F}_{q}\}\leq\mathbb{F}_{q}^{n}$.\\
Furthermore, given a permutation $\sigma$ in the symmetric group $S_{n}$, we denote, with abuse of notation, the $\mathbb{F}_{q}$-isomorphism $\sigma\colon\mathbb{F}_{q}^{n}\to\mathbb{F}_{q}^{n}$ associated to the permutation $\sigma$ such that $\sigma(\sum_{i=1}^{n}a_{i}e_{i})=\sum_{i=1}^{n}a_{i}e_{\sigma(i)}$, for any $\sum_{i=1}^{n}a_{i}e_{i}\in\mathbb{F}_{q}^{n}$. In particular, we say that two linear codes $C,C'\leq\mathbb{F}_{q}^{n}$ are {\it permutation equivalent} if $\sigma(C)=C'$ for some $\sigma\in S_{n}$. Moreover, following the notation used in \cite{BRS09}, for an arbitrary linear code $C\leq\mathbb{F}_{q}^{n}$ we denote by $PAut(C)=\{\sigma\in S_{n}\ |\ \sigma(C)=C\}$ the permutation automorphisms group of $C$.\\
Finally, we recall that a subgroup $G$ of $S_{n}$ is said to be a {\it regular subgroup} if it has order $n$ and its natural action on the set $\{1,\ldots,n\}$ is transitive.

\begin{lem}\label{lem:ai}
Let $H$ be a regular subgroup of $S_{n}$ and fix an arbitrary element $i_{0}\in\{1,\ldots,n\}$. Let $\psi\colon H\to\{1,\ldots,n\}$ be the bijection given by $\psi(h)=h(i_{0})$, for every $h\in H$. Then there is an anti-isomorphism $\sigma\colon H\to C_{S_{n}}(H)$ which maps $h\in H$ to $\sigma_{h}$, where for any $i\in\{1,\ldots,n\}$
\[
\sigma_{h}(i)=\psi^{-1}(i)(h(i_{0}))
\]
Moreover $\sigma_{h}=h$ for every $h\in Z(H)$ and hence $Z(H)=Z(C_{S_{n}})$.
\end{lem}
\begin{proof}
See for instance \cite[Lemma 1.1]{BRS09}.
\end{proof}

We are going now to outline a necessary and sufficient condition for a linear code to be a (left) group code.

\begin{thm}[\cite{BRS09}, Theorem 1.2]\label{thm:nscagc}
Let $C\leq\mathbb{F}_{q}^{n}$ be a linear code and $G$ be a finite group of order $n$.
\begin{itemize}
\item[i)]{$C$ is a left $G$-code if and only if $G$ is isomorphic to a regular subgroup of $S_{n}$ contained in $PAut(C)$.}
\item[ii)]{$C$ is a $G$-code if and only if $G$ is isomorphic to a regular subgroup $H$ of $S_{n}$ such that $\langle H,C_{S_{n}}(H)\rangle\leq PAut(C)$.}
\end{itemize}
\end{thm}

\begin{cor}[\cite{BRS09}, Corollary 1.3]\label{cor:nscgc}
Let $C\leq\mathbb{F}_{q}^{n}$ be a linear code and $\mathcal{G}$ be a class of groups.
\begin{itemize}
\item[i)]{$C$ is a left $\mathcal{G}$-code if and only if PAut($C$) contains a regular subgroup of $S_{n}$ in $\mathcal{G}$.}
\item[ii)]{$C$ is a $\mathcal{G}$-code if and only if PAut($C$) contains a regular subgroup $H$ of $S_{n}$ in $\mathcal{G}$ such that $C_{S_{n}}(H)\leq$ PAut($C$).}
\end{itemize}
\end{cor}

As a direct consequence of Corollary~\ref{cor:nscgc}, Bernal, del Río and Simon have proved \cite[Theorem 3.1]{BRS09} that for any finite group $G$, having an abelian decomposition, every $G$-code is an abelian group code. We recall in particular that a group $G$, not necessary finite, has an {\it abelian decomposition} if there exist abelian subgroups $A$ and $B$ of $G$ such that $G=AB=\{ab\ |\ a\in A,b\in B\}$.\\
We get the following Theorem due to Bernal, del Río and Simon \cite{BRS09}, which will be a powerful and essential  preliminary for the following results of this section. It should be noted that we are going to show the proof of the following Theorem~\ref{thm:bdsad} because proof of Theorem~\ref{thm:ccdGccgc} will use quite similar arguments and techniques. 

\begin{thm}[\cite{BRS09}, Theorem 3.1]\label{thm:bdsad}
Let $C\le\mathbb{F}_{q}^{n}$ be a linear code. Suppose there exists a regular subgroup $G$ of $S_{n}$ having an abelian decomposition and such that $\langle G,C_{S_{n}}(G)\rangle\leq PAut(C)$. Then, $C$ is an abelian group code.
\end{thm}
\begin{proof}
Let $A,B\leq G$ be abelian subgroups such that $G=AB$ and let $\sigma\colon G\to C_{S_{n}}(G)$ be the anti-isomorphism as in Lemma~\ref{lem:ai}, with respect to some $i_{0}\in\{1,\ldots,n\}$, and set $A_{1}=\sigma(A)$ and $B_{1}=\sigma(B)$.\\
Considering the group $K=\langle A,B_{1}\rangle\leq\langle G,C_{S_{n}}(G)\rangle\leq$ PAut($C$), we can observe that $K$ is an abelian group, as $A$ and $B_{1}$ are abelian subgroups such that $A\leq G$ and $B_{1}\leq C_{S_{n}}(G)$. In order to show that $C$ is a $K$-code, and therefore an abelian group code, it is enough to prove, by Theorem~\ref{thm:nscagc}, that $|K|=n$ and the stabilizier in $K$ of every element of $\{1,\ldots,n\}$ is trivial. \\
To prove the first condition, by Lemma~\ref{lem:ai} we have $B\cap Z(G)=\sigma(B\cap Z(G))=B_{1}\cap Z(C_{S_{n}}(G))$, thus $[B:B\cap Z(G)]=[B_{1}:B_{1}\cap Z(C_{S_{n}}(G))]$ and $[B\cap Z(G):B\cap A]=[B_{1}\cap Z(C_{S_{n}}(G)):B_{1}\cap A]$. It implies that $[G:A]=[AB:A]=[B:B\cap A]=[B:B\cap Z(G)][B\cap Z(G):A]=[B_{1}:B_{1}\cap Z(C_{S_{n}}(G))][B_{1}\cap Z(C_{S_{n}}(G)):B_{1}\cap A]=[B_{1}:B_{1}\cap A]=[AB_{1}:A]=[K:A]$, and so $|K|=|G|=n$.\\
Finally, let $i\in\{1,\ldots,n\}$ and $k\in K$ such that $k(i)=i$. Then $k=a\beta$ for some $a\in A$ and $\beta=\sigma_{b}\in B_{1}$, with $b\in B$. Furthermore, following the notation used in Lemma~\ref{lem:ai} and being $G$ a regular subgroup of $S_{n}$, then there exists a unique $g\in G$ such that $g(i_{0})=i$. Therefore $g(i_{0})=i=k(i)=a\beta g(i_{0})=a\sigma_{b}(\psi(g))=abg(i_{0})$, thus $agb=g$. Setting $g=a'b'$, with $a'\in A$ and $b'\in B$, then $a'abb'=aa'b'b=agb=g=a'b'$ and hence  $ab=1_{G}$. We conclude observing that, being $b^{-1}=a\in A\cap B\leq Z(G)$, then $\beta=\sigma_{b}=b$ and hence $k=a\beta=ab=1_{G}$.
\end{proof}

Following the terminology used in \cite{PGMMM19}, given two groups $G$ and $H$, we say that the left (right, two-sided) ideals $I$ and $J$ of the group algebras $\mathbb{F}_{q}[G]$ and $\mathbb{F}_{q}[H]$, respectively, are {\it permutation equivalent} if there exists a bijection $\varphi\colon G\to H$ such that $\tilde\varphi(I)=J$, where $\tilde\varphi$ denotes the $\mathbb{F}_{q}$-linear extension of $\varphi$ to the corresponding group algebras. It should be observed that this $\mathbb{F}_{q}$-linear extension $\tilde{\varphi}$ is not necessary an automorphism of $\mathbb{F}_{q}$-algebras because bijection $\varphi$ is not required to be an isomorphism of groups. Furthermore, given a linear code $C\leq\mathbb{F}_{q}^{n}$ and a finite group $G$ of order $n$, we remember that the linear code $C$ is {\it permutation equivalent} to some (left, right, two-sided) ideal $I$ of the group algebra $\mathbb{F}_{q}[G]$ if there exists a bijection $\phi\colon\mathcal{B}\to G$ such that $\tilde{\phi}(C)=I$.\\
Finally, we say that a subgroup $U$ of $G$ {\it acts trivially} from the left (from the right) on some subset $X\subseteq\mathbb{F}_{q}[G]$ if $u\cdot x=x$ (respectively, $x\cdot u=x$) for every $u\in U$ and $x\in X$.\\
We get the following results due to Pillado et al., \cite{PGMMM19}.

\begin{lem}[\cite{PGMMM19}, Lemma 1.2]\label{lem:p1}
Let $G$ and $H$ two groups of the same order $n$. Suppose that there exist two normal subgroups $N\unlhd G$ and $K\unlhd H$ such that $G/N\cong H/K$. If $N$ acts trivially on some (left, right, two-sided) ideal $I$ of $\mathbb{F}_{q}[G]$, then $I$ is permutation equivalent to some (respectively; left, right, two-sided) ideal of the group algebra $\mathbb{F}_{q}[H]$.
\end{lem}
\begin{lem}[\cite{PGMMM19}, Lemma 1.3]\label{lem:p2}
Let $G$ be a group and suppose that a normal subgroup $N\unlhd G$ acts trivially from the left (from the right) on some (left, right, two-sided) ideal $I$ of the group algebra $\mathbb{F}_{q}[G]$ and that $G/N$ has an abelian decomposition. Then $I$ is permutation equivalent to an ideal of the group algebra $\mathbb{F}_{q}[A]$ for some abelian group $A$.
\end{lem}

An immediate consequence of Lemma~\ref{lem:p2} concerns the following necessary and sufficient condition for those ideals of the group algebra $\mathbb{F}_{q}[G]$ which are permutation equivalent to some ideals of abelian group algebras.

\begin{cor}[\cite{PGMMM19}, Corollary 1.4]\label{cor:nscagc2}
Let $G$ be a finite group and $I$ be a (left, right, two-sided) ideal of the group algebra $\mathbb{F}_{q}[G]$. Then, $I$ is permutation equivalent to an ideal of the group algebra $\mathbb{F}_{q}[A]$, for some abelian group $A$, if and only if the derived subgroup $G'$ acts trivially from the left (from the right) on the ideal $I$.
\end{cor}

We are going now to reformulate Corollary~\ref{cor:nscagc2} with respect to the group of permutation automorphisms of those linear codes which are abelian group codes.

\begin{cor}\label{cor:nscabgdsat}
Let $C\leq\mathbb{F}_{q}^{n}$ be a linear code. Then $C$ is an abelian group code if and only if there exists a regular subgroup $G$ of $S_{n}$ such that $\langle G,C_{S_{n}}(G)\rangle\leq PAut(C)$ and such that the linear code $C$ is permutation equivalent to some (left, right, two-sided) ideal of the group algebra $\mathbb{F}_{q}[G]$, on which the derived subgroup $G'$ acts trivially from the left (from the right).
\end{cor}

The previous Corollary~\ref{cor:nscabgdsat} suggests us to establish a condition on the minimum Hamming weight of some special abelian group codes. Given an integer $\Delta>1$, we recall that a linear code $C\leq\mathbb{F}_{q}^{n}$ is said to be a $\Delta$-divisible code if $\Delta\mid w_{H}(c)$, for every codeword $c=(c_{1},\ldots,c_{n})\in C$, where $w_{H}(c):=|\{i\in\{1,\ldots,n\}\ |\ c_{i}\neq0\}|$ (See for instance \cite{H98}). Furthermore, we recall that the minimum Hamming weight of a linear code $C$ is defined as $w(C):=\min\{w_{H}(c)\ |\ 0\neq c\in C\}$.\\
We are now in position to prove Theorem~\ref{thm1}.

\begin{proof}[Proof of Theorem~\ref{thm1}]
By Theorem~\ref{thm:nscagc}, let $\phi\colon\mathcal{B}\to G$ be a bijection such that its $\mathbb{F}_{q}$-linear extension $\tilde\phi\colon\mathbb{F}_{q}^{n}\to\mathbb{F}_{q}[G]$ is such that $\tilde\phi(C)$ is a two-sided ideal of the group algebra $\mathbb{F}_{q}[G]$ on which the derived subgroup $G'$ acts trivially from the left (from the right). Let $c=(c_{1},\ldots,c_{n})\in C$ be an arbitrary codeword and $g\in G'$. Set $z=\tilde\phi(c)=\sum_{i=1}^{n}c_{i}\phi(e_{i})$ and fix an index $i\in\{1,\ldots,n\}$. Since $G'$ acts trivially from the left on $\tilde\phi(C)$, then $gz=z$. Denoting by $i^{(g)}\in\{1,\ldots,n\}$ the index for which $\phi(e_{i^{(g)}})=g^{-1}\phi(e_{i})$, then, being $z=\sum_{s=1}^{n}c_{s}\phi(e_{s})=gz=\sum_{s=1}^{n}c_{s}g\phi(e_{s})=c_{i^{(g)}}\phi(e_{i})+\sum_{s\neq i}c_{s}g\phi(e_{s})$, we have $c_{i}=c_{i^{(g)}}$. Set $S_{i}=\{i^{(g)}\ |\ g\in G'\}\subseteq\{1,\ldots,n\}$. Thus $|\{e_{i^{(g)}}\ |\ g\in G'\}|=|\{\phi^{-1}(g^{-1}\phi(e_{i}))\ | \ g\in G'\}|=|G'|$, hence $|S_{i}|=|G'|$ and $c_{i}=c_{i^{(g)}}$, for every $g\in G'$. Consider the natural action $G'\times G\to G$ given by the multiplication from the left by elements of $G'$. Since $G=\bigsqcup_{i\in J}O_{G}(\phi(e_{i}))$, for a suitable complete set of representatives $J\subseteq\{1,\ldots,n\}$ for the $G'$-orbits of $G$, and the stabilizer in $G'$ of every element of $G$ is trivial, we have $|O_{G}(\phi(e_{i}))|=|G'|$, for every $i\in\{1,\ldots,n\}$, and hence $S_{i}\cap S_{j}=\emptyset$, for any $i, j\in J$ such that $i\neq j$. Finally, since $G'\neq\{1\}$, we get that $C$ is a $|G'|$-divisible code and in particular $|G'|\mid w(C)$.
\end{proof}

As a direct consequence of Theorem~\ref{thm1}, we can provide an equivalent description of these above abelian group codes in terms of subspaces of some specified linear codes, Corollary~\ref{thm2}.

\begin{proof}[Proof of Corollary~\ref{thm2}]
By Theorem~\ref{thm:nscagc} $C$ is permutation equivalent to some two-sided ideal of the group algebra $\mathbb{F}_{q}[G]$ on which the derived subgroup $G'$ acts trivially from the left (from the right). Suppose that $|G'|=t>1$ and $[G:G']=s>1$. By using the same notation as in Theorem~\ref{thm1}, for every $j\in J$ the $G'$-orbits of $G$ have cardinality given by $|O_{G}(\phi(e_{j}))|=t$. Let ${\bf0}\neq c=(c_{1},\ldots,c_{n})$ an arbitrary non-zero codeword of $C$ and, with abuse of notation, write this in form $c=\sum_{k=1}^{n}c_{k}\phi(e_{k})$. Fix an index $j\in J$. Thus, the component $c_{k}$ equals to $c_{j}$, for every index $k$ such that $\phi(e_{k})\in O_{G}(\phi(e_{j}))$. It implies that, up to a suitable permutation of the coordinates of $\mathbb{F}_{q}^{n}$, the abelian group code $C$ is equivalent to some subspace of the direct sum of $s$ copies of the repetition code $Rep_{t}(\mathbb{F}_{q})$.
\end{proof}

\section{The Direct Sum of $s$ Copies of the Repetition Code of Length $t$ is an Abelian Group Code}\label{Section2}

Starting from the characterisation of a restricted family of abelian group codes provided in Corollary~\ref{thm2}, we are now going to wonder ourselves if the extremal codes, namely the direct sums of $s$ copies of the repetition code $Rep_{t}(\mathbb{F}_{q})$, can be realised as abelian group codes. In this regard, in the following Theorem we give a complete description of all permutation automorphisms which fix the linear code $C=\bigoplus_{i=1}^{s}Rep_{t}(\mathbb{F}_{q})$ and, in force of Theorem~\ref{thm:nscagc}, we establish that the linear code $C$ is an abelian group code for every $s,t\geq1$. It is easy to observe that we may restrict our attention only to the case in which both $s$ and $t$ are strictly greater than $1$; indeed, $s=1$ or $t=1$ correspond, respectively, to the cyclic group codes $C=Rep_{t}(\mathbb{F}_{q})$ and $C=\mathbb{F}_{q}^{s}$.
\begin{proof}[Proof of Theorem~\ref{thm3}]
Fix two positive integers $s,t>1$. In order to prove the theorem we have to show, by virtue of Theorem~\ref{thm:nscagc}, that there exists an abelian regular subgroup of $S_{st}$ contained in $PAut(C)$. Note that $PAut(Rep_{t}(\mathbb{F}_{q}))=S_{t}$. For ease of notation, we index the components of an arbitrary codeword $c$ of the linear code $C$ in the following way
\[
c=(c_{1^{(1)}},\ldots,c_{t^{(1)}},c_{1^{(2)}},\ldots,c_{t^{(2)}},\ldots,c_{1^{(s)}},\ldots,c_{t^{(s)}}).
\]
Let $\pi\in PAut(C)$ and fix an index $i\in\{1,\ldots,s\}$. Since, for every $j\in\{1,\ldots,s\}$, $c_{h^{(j)}}=c_{k^{(j)}}$ for any $h,k\in\{1,\ldots,t\}$, the block of components $B_{i}:=\{c_{1^{(i)}},\ldots,c_{t^{(i)}}\}$ is mapped bijectively by $\pi$ to the block of components $B_{j}=\{c_{1^{(j)}},\ldots,c_{t^{(j)}}\}$, for some (unique) index $j\in\{1,\ldots,s\}$. This implies that there exists a permutation $\sigma_{\pi}\in S_{s}$ such that, for every $i\in\{1\ldots,s\}$, $\sigma_{\pi}(i)=j$ if and only if the block of components $B_{i}$ is mapped by $\pi$ to the block of components $B_{j}$. Since $B_{i}$ is mapped bijectively by $\pi$ in $B_{\sigma_{\pi}(i)}$, then there exists a (unique) permutation $\tau_{i,\sigma_{\pi}(i)}\in S_{t}$ such that $c_{\tau_{i,\sigma_{\pi}(i)}(h)^{(i)}}=c_{\pi(h^{(i)})}$, for every $h\in\{1,\ldots,t\}$. It follows that the permutation $\pi\in PAut(C)$ is completely and uniquely determined by the permutation $\sigma_{\pi}\in S_{s}$ and the set of permutations $\{\tau_{1,\sigma_{\pi}(1)},\ldots,\tau_{s,\sigma_{\pi}(s)}\}\subseteq S_{t}$.\\
Let $\pi_{1}\in PAut(C)$ such that $\sigma_{\pi_{1}}$ is the $s$-cycle $(12\cdots s)\in S_{s}$ and $\tau_{i,\sigma_{\pi_{1}}(i)}=\iota$, for every $i\in\{1,\ldots,s\}$, and let $\pi_{2}\in PAut(C)$ such that $\sigma_{\pi_{2}}=\iota$ and $\tau_{i,\sigma_{\pi_{2}}(i)}$ is the $t$-cycle $(12\cdots t)\in S_{t}$, for every $i\in\{1,\ldots,s\}$. Let $G=A_{1}A_{2}$, where $A_{1}=\langle\pi_{1}\rangle\leq PAut(C)$ and $A_{2}=\langle\pi_{2}\rangle\leq PAut(C)$. Since, by construction, $A_{1}$ and $A_{2}$ are cyclic subgroups of $PAut(C)$ of order, respectively, $s$ and $t$, $A_{1}\cap A_{2}=\{\iota\}$ and $\pi_{1}\pi_{2}=\pi_{2}\pi_{1}$, then we have that $G$ is an abelian subgroup of $S_{st}$ such that $|G|=st$ and $G$ is contained in $PAut(C)$. Moreover, it is easy to see that $G$ is a regular subgroup. Consider the index $1^{(1)}$ and let $i\in\{0,\ldots,s-1\}$ and $j\in\{0,\ldots,t-1\}$; then, the permutation $\pi_{1}^{i}\pi_{2}^{j}$ maps index $1^{(1)}$ to the index $(j+1)^{(i+1)}$. Finally, by Theorem~\ref{thm:nscagc}, $C$ is a $G$-code and hence an abelian group code.
\end{proof}

\begin{cor}
Let $s,t\geq1$ be coprime integers. Then, $C=\bigoplus_{i=1}^{s}Rep_{t}(\mathbb{F}_{q})\leq\mathbb{F}_{q}^{n}$ is a cyclic group code.
\end{cor}
\begin{proof}
It follows directly from the fact that the group $G$, as in proof of Theorem~\ref{thm3}, is a cyclic group because $G\cong A_{1}\times A_{2}\cong\mathcal{C}_{s}\times\mathcal{C}_{t}\cong\mathcal{C}_{st}$ is isomorphic to the cyclic group of order $s\cdot t$.
\end{proof}

In what follows we would explain an alternative way of answering the above question and, since it won't provide a complete answer, we point out that we are showing it for completeness reasons as the tools which we have used are strictly related to the theory of group codes. Specifically, by using idea of Corollary~\ref{cor:nscabgdsat}, we partially provide in Proposition~\ref{prop1} a condition for which the direct sums of $s$ copies of the repetition code $Rep_{t}(\mathbb{F}_{q})$ can be realised as abelian group codes. As a result, it is easy to find many infinite pairs of positive integers $s$ and $t$ for which hypothesis of Proposition~\ref{prop1} are satisfied.\\
Before proving Proposition~\ref{prop1}, we recall that for a finite group $G$ and a normal subgroup $N\unlhd G$, the element $N_{\Sigma}:=\sum_{n\in N}n\in\mathbb{F}_{q}[G]$ is such that $nN_{\Sigma}=N_{\Sigma}=N_{\Sigma}n$, for every $n\in N$, and $N_{\Sigma}\in Z(\mathbb{F}_{q	}[G])$, where $Z(\mathbb{F}_{q	}[G])$ denotes the centre of the group algebra $\mathbb{F}_{q}[G]$.

\begin{prop}\label{prop1}
Let $s$ and $t$ be positive integers. If there exists a finite group $G$ of order $s\cdot t=|G|$ and such that $|G'|=t$, then the linear code $C=\bigoplus_{i=1}^{s}Rep_{t}(\mathbb{F}_{q})$ is an abelian group code.
\end{prop}
\begin{proof}
Consider the linear code $C=\bigoplus_{i=1}^{s}Rep_{t}(\mathbb{F}_{q})\leq\mathbb{F}_{q}^{st}$ and denote by $\mathcal{A}=\{f_{1},\ldots,f_{s}\}$ a basis for $C$, where $f_{i}=\sum_{r=(i-1)t+1}^{it}e_{r}$ for every $i\in\{1,\ldots,s\}$. It is easy to see that $dim_{\mathbb{F}_{q}}(C)=s$, $w(C)=t$ and $C$ is a $t$-divisible code. Let $\mathcal{T}=\{g_{1},\ldots,g_{s}\}\subseteq G$ be a complete set of representatives for the left cosets in $G$ module $G'$ and let $G'=\{h_{1}=1_{G},\ldots,h_{t}\}$. We define the function $\phi\colon\mathcal{B}\to G$ as follows. Let $r\in\{1,\ldots,st\}$ and consider the (unique) integer $j\in\{1,\ldots,s\}$ such that $(j-1)t+1\leq r\leq jt$; thus, we define $\phi(e_{r})=g_{j}h_{r-(j-1)t}$. It is easy to see that, by construction, $\phi$ is a bijection and its $\mathbb{F}_{q}$-linear extension is such that $\tilde\phi(f_{i})=g_{i}(h_{1}+\ldots+h_{t})=g_{i}(G')_{\Sigma}$, for every $i=1,\ldots,s$. Furthermore, after fixing an index $i\in\{1,\ldots,s\}$ and choosing an element $g=g_{u}h_{v}\in G$ with $g_{u}\in\mathcal{T}$ and $h_{v}\in G'$, we have $g\tilde\phi(f_{i})=g_{u}h_{v}g_{i}(G')_{\Sigma}=g_{u}h_{v}(G')_{\Sigma}g_{i}=g_{u}(G')_{\Sigma}g_{i}=g_{u}g_{i}(G')_{\Sigma}=g_{w}h_{w}(G')_{\Sigma}=\tilde\phi(f_{w})\in\tilde\phi(C)$ for some $w\in\{1,\ldots,s\}$. In a similar fashion one sees that there exists an index $m\in\{1,\ldots,s\}$ such that $\tilde\phi(f_{i})g=\tilde\phi(f_{m})\in\tilde\phi(C)$. Thus, by $\mathbb{F}_{q}$-linearity, $\tilde\phi(C)$ is a two-sided ideal $\mathbb{F}_{q}[G]$. In addiction, for every index $v\in\{1,\ldots,t\}$, we have $h_{v}\tilde\phi(f_{i})=h_{v}g_{i}(G')_{\Sigma}=h_{v}(G')_{\Sigma}g_{i}=(G')_{\Sigma}g_{i}=g_{i}(G')_{\Sigma}=\tilde\phi(f_{i})$, that is the derived subgroup $G'$ acts trivially from the left on the ideal $\tilde\phi(C)$. Finally, by Corollary~\ref{cor:nscagc2}, $\tilde\phi(C)$ is permutation equivalent to some ideal of an abelian group algebra and hence $C$ is an abelian group code.
\end{proof}

The previous Proposition~\ref{prop1} suggests us considering the following question.
\begin{que}\label{que:1}
Is there exists, for every pair of integers $s,t>1$, a finite group $G$ such that the derived subgroup has order $t$ and index $s$?
\end{que}

Obviously, if Question~\ref{que:1} will have a positive answer then, by Proposition~\ref{prop1}, we can conclude again that the linear code $C=\bigoplus_{i=1}^{s}Rep_{t}(\mathbb{F}_{q})$ is an abelian group code, for every choices of the integers $s,t>1$. In particular with the following propositions we try, at least partially, to answer the Question~\ref{que:1}.

\begin{prop}\label{prop:1section2}
Let $s,t\in\N$ be positive integers and consider the linear code $C=\bigoplus_{i=1}^{s}Rep_{t}(\mathbb{F}_{q})$.
\begin{itemize}
\item[i)]{For every $t>1$ such that there exists a perfect group of order $t$, then $C$ is an abelian group code for every $s\geq1$.}
\item[ii)]{For every odd $t$ and even $s$, $C$ is an abelian group code.}
\item[iii)]{For every $t\geq1$ and $s\equiv_{4}0$, $C$ is an abelian group code.}
\end{itemize}
\end{prop}
\begin{proof}
On our way to proving this result, we construct certain finite groups whose derived subgroups have prescribed order and index, for some special choices of the integers $s$ and $t$.\\
$i)$\ Let $t>1$ an integer such that there exists a perfect group $G$ of order $t$. Consider the group $H=G\times\mathcal{C}_{s}$ given by the direct product of $G$ with the cyclic group $\mathcal{C}_{s}$ of order $s$. Then, since by construction $|H|=ts$ and $H'=G'\times(\mathcal{C}_{s})'=G\times\{0\}$, we have $|H'|=t$ e $|H/H'|=s$. By Proposition~\ref{prop1}, $C$ is an abelian group code.\\
$ii)$\ Let $t$ an odd integer and consider the dihedral group $D_{2t}=\langle x,y \ | \ x^{2}=1=y^{t}, yxy=x\rangle$. Since $(D_{2t})'=\langle y^{2}\rangle=\langle y\rangle$, for every integer $r\geq1$ the group $H=D_{2t}\times\mathcal{C}_{r}$ is such that $|H|=2tr$, $|H'|=|\langle y\rangle\times\{0\}|=t$ and $[H:H']=2r=:s$. Again, we can conclude by Proposition~\ref{prop1}.\\
$iii)$\ In a similar fashion as in $ii)$, considering the fact that for every even integer $k=2t$ we have $(D_{2(2t)})'=\langle y^{2}\rangle$ is a proper subgroup of $\langle y\rangle$, hence $|(D_{2(2t)})'|=t$ and $[D_{2(2t)}:(D_{2(2t)})']=4$.
\end{proof}

A special class of finite groups, containing properly the class of dihedral groups, can be used to extend the reasoning used in $ii)$ and $iii)$ of Proposition~\ref{prop:1section2}. \\
Let $p$ and $q$ prime numbers such that $p\mid q-1$ and $m\in\N$ such that $m\not\equiv_{q}1$ and $m^{p}\equiv_{q}1$. Consider the finite group $G_{p,q,m}$ given by:
\[
G_{p,q,m}:=\langle\alpha,\beta\ |\ \alpha^{p}=1=\beta^{q}, \alpha\beta\alpha^{-1}=\beta^{m}\rangle.
\]
First of all, we recall that the condition $p\mid q-1$ implies the existence of a such $m$, as $\mathbb{F}_{q}^{\ast}$ is a cyclic group of order $q-1$. Moreover, $G_{p,q,m}=\langle\beta\rangle\rtimes\langle\alpha\rangle\cong\mathcal{C}_{q}\rtimes\mathcal{C}_{p}$ is a non abelian split-metacyclic group: indeed $\langle\beta\rangle\unlhd G_{p,q,m}$ is a normal cyclic subgroup of order $q$, that is the unique $q$-Sylow subgroup of $G_{p,q,m}$, such that $G_{p,q,m}/\langle\beta\rangle\cong\langle\alpha\rangle\cong\mathcal{C}_{p}$. In particular, being $G_{p,q,m}$ a non abelian group and $G_{p,q,m}/\langle\beta\rangle$ a cyclic group, it follows that $(G_{p,q,m})'=\langle\beta\rangle$, hence $|(G_{p,q,m})'|=q$ and $[G_{p,q,m}:(G_{p,q,m})']=p$.

\begin{prop}
Let $s=p_{1}^{\gamma_{1}}\cdots p_{r}^{\gamma_{r}}$ be a positive integer, with $p_{i}$ distinct primes and $\gamma_{i}\in\N\setminus\{0\}$. Let $t\in\N_{>1}$ be an integer, where $t=q_{1}^{\delta_{1}}\cdots q_{l}^{\delta_{l}}$ with $l\leq r$, $q_{i}$ distinct primes such that, up to a suitable renumbering, $p_{i}\mid q_{i}-1$ and $\delta_{i}\in\N\setminus\{0\}$ such that $\delta_{i}\leq\gamma_{i}$, for every $i=1,\ldots,l$. Then the linear code $C=\bigoplus_{i=1}^{s}Rep_{t}(\mathbb{F}_{q})$ is an abelian group code.
\end{prop}
\begin{proof}
Fix an index $i\in\{1,\ldots,l\}$ and an integer $m_{i}$ such that $m_{i}\not\equiv_{q_{i}}1$ e $m_{i}^{p_{i}}\equiv_{q_{i}}1$. Consider the group $G_{p_{i},q_{i},m_{i}}=\langle\alpha_{i},\beta_{i}\ |\ \alpha_{i}^{p_{i}}=1=\beta_{i}^{q_{i}}, \alpha_{i}\beta_{i}\alpha_{i}^{-1}=\beta_{i}^{m_{i}}\rangle=\langle\beta_{i}\rangle\rtimes\langle\alpha_{i}\rangle$, where $|(G_{p_{i},q_{i},m_{i}})'|=q_{i}$ and $[G_{p_{i},q_{i},m_{i}}:(G_{p_{i},q_{i},m_{i}})']=p_{i}$. We are going to construct the finite group $G_{i}:=(G_{p_{i},q_{i},m_{i}}\times\cdots\times G_{p_{i},q_{i},m_{i}})\times\mathcal{C}_{p_{i}^{(\gamma_{i}-\delta_{i})}}$ given by the direct product of $\delta_{i}$ copies of the group $G_{p_{i},q_{i},m_{i}}$ with the cyclic group $\mathcal{C}_{p_{i}^{(\gamma_{i}-\delta_{i})}}$ of order $p_{i}^{(\gamma_{i}-\delta_{i})}$. Thus, by construction, $|G_{i}'|=q_{i}^{\delta_{i}}$ and $[G_{i}:G_{i}']=p_{i}^{\delta_{i}}p_{i}^{(\gamma_{i}-\delta_{i})}=p_{i}^{\gamma_{i}}$. Providing this construction for the finite group $G_{i}$ for every index $i\in\{1,\ldots,l\}$, we are now in position to consider the finite group $G:=(G_{1}\times\cdots\times G_{l})\times\mathcal{C}$, where $\mathcal{C}$ is the cyclic group of order $p_{l+1}^{\gamma_{l+1}}\cdots p_{r}^{\gamma_{r}}$. Finally, as $|G'|=t$ and $[G:G']=s$, the linear code $C$ is an abelian group code by Proposition~\ref{prop1}.
\end{proof}
We may observe that Dirichlet's Theorem on arithmetic progressions for prime numbers (See for instance \cite{S49}) ensures that, for a fixed positive integer $s$, there exist many infinite integers $t$ such that the linear code $C=\bigoplus_{i=1}^{s}Rep_{t}(\mathbb{F}_{q})$ is an abelian group code: indeed, for every index $i\in\{1,\ldots,l\}$, there exist many infinite prime numbers $q_{i}$ such that $p_{i}\mid q_{i}-1$.\\

\section{An Equivalent Characterisation of Cyclic Group Codes}\label{Section3}
In this section we give an equivalent description of cyclic group codes, adopting a similar approach used in  \cite{PGMMM19} by Pillado et al. for abelian group codes. In particular, our main result is based on the following Theorem~\ref{thm:ccdGccgc}, which appears as special case of Theorem~\ref{thm:bdsad}. Moreover, we will see that not only does Theorem~\ref{thm:ccdGccgc} give rise to new several examples of cyclic group codes, but also gives rise to a different way of describing cyclic group codes, as stated in Theorem~\ref{cor:nsccgc3}.\\
In this regard, we recall that a finite group $G$ has a {\it coprime cyclic decomposition} if there exist cyclic subgroups $A,B\leq G$ of coprime order such that $G=AB$.

\begin{thm}\label{thm:ccdGccgc}
Let $C\le\mathbb{F}_{q}^{n}$ be a linear code. Suppose there exists a regular subgroup $G$ of $S_{n}$ having a coprime cyclic decomposition and such that $\langle G,C_{S_{n}}(G)\rangle\leq PAut(C)$. Then, $C$ is a cyclic group code.
\end{thm}
\begin{proof}
The proof works exactly in the same way as for Theorem~\ref{thm:bdsad}. Suppose $G=AB$, where $A$ and $B$ are cyclic subgroups of $G$ of coprime order and let $\sigma\colon G\to C_{S_{n}}(G)$ be the anti-isomorphism as in Lemma~\ref{lem:ai}, with respect to some $i_{0}\in\{1,\ldots,n\}$. Set $A_{1}=\sigma(A)$ and $B_{1}=\sigma(B)$ and consider the group $K=\langle A,B_{1}\rangle\leq\langle G,C_{S_{n}}(G)\rangle\leq PAut(C)$. To prove the theorem it suffices to show that $K$ is a cyclic group. Indeed, since $A$ and $B_{1}$ are cyclic subgroups such that $A\leq G$ and $B_{1}\leq C_{S_{n}}(G)$, then $K$ is an abelian group. Furthermore, having $A$ and $B_{1}$ coprime order, and so $A\cap B_{1}=\{1\}$, then $K\cong A\times B_{1}$ and hence $K$ is a cyclic group of order $|A||B_{1}|$.\\
We conclude by noting that the proof of regularity of $K$ as a subgroup of $S_{n}$ works in the same manner as for Theorem~\ref{thm:bdsad}.
\end{proof}

Before explaining the following remark, we recall that a group $G$ is metacyclic if there exist cyclic subgroups $K,S\leq G$, with $K$ normal in $G$, such that $G=SK$; in the special case in which $S\cap K=\{1\}$, then $G=K\rtimes S$ is said to be split-metacyclic. 

\begin{rmk}
An immediate consequence of Theorem~\ref{thm:ccdGccgc} concerns the determination of a wide class of group codes which can be realised as cyclic group codes. In \cite{SL95} Sabin and Lomonaco have proved that if $C$ is a $H$-code for a split-metacyclic group $H$ then $C$ is an abelian group code. In \cite{BRS09}, Bernal, Del Río and Simon have extended this result to arbitrary metacyclic group codes. In particular, we are going to refine these results showing that under some specified conditions on the split-metacyclic group $H$ then $C$ can be realised as a cyclic group code. Thus, suppose $H=SK$ is a metacyclic group with $S\leq H$ and $K\unlhd H$ cyclic subgroups of coprime order. If $C\leq\mathbb{F}_{q}^{n}$ is a linear code satisfying hypothesis of Theorem~\ref{thm:ccdGccgc} with $G\cong H$, then $C$ is a cyclic group code. As a direct consequence of this observation, we may observe that every $D_{2m}$-dihedral group code is a cyclic group code, whenever $m$ is an odd integer.
\end{rmk}

One of the most significant consequence of Theorem~\ref{thm:ccdGccgc} concerns the following equivalent characterisation of cyclic group codes, Theorem~\ref{cor:nsccgc3}. In a similar fashion as for abelian group codes \cite{PGMMM19}, we need the following lemma, which appears as a special case of Lemma~\ref{lem:p2}. In particular, we recall that, given a group $G$, a non-trivial normal subgroup $N\unlhd G$ is said to be a {\it co-cyclic subgroup} of $G$, if the quotient group $G/N$ is cyclic.

\begin{lem}\label{lem:Hnccspei}
Let $G$ be a group and suppose that an Hall normal subgroup $N\unlhd G$ acts trivially from the left (from the right) on some (left, right, two-sided) ideal $I$ of the group algebra $\mathbb{F}_{q}[G]$ and that $G/N$ has a coprime cyclic decomposition. Then $I$ is permutation equivalent to an ideal of a group algebra $\mathbb{F}_{q}[\mathcal{C}]$ for some cyclic group $\mathcal{C}$.
\end{lem}
\begin{proof}
Consider the group $H=(G/N)\times \mathcal{C}_{|N|}$, where $\mathcal{C}_{|N|}$ is a cyclic group of order $|N|$. Thus, $H$ is a finite group of order $n$ such that $\mathcal{C}_{|N|}\unlhd H$ and $H/\mathcal{C}_{|N|}\cong G/N$. By Lemma~\ref{lem:p1} $I$ is permutation equivalent to some ideal $J$ of $\mathbb{F}_{q}[H]$. Finally, as $G/N$ has a coprime cyclic decomposition and $gcd(|N|,|G/N|)=1$, $H$ as a coprime cyclic decomposition too, and, by Theorem~\ref{thm:ccdGccgc}, we can conclude.
\end{proof}

Finally, we are now in position to prove Theorem~\ref{cor:nsccgc3}.
\begin{proof}[Proof of Theorem~\ref{cor:nsccgc3}]
By Theorem~\ref{thm:nscagc}, $C$ is permutation equivalent to some ideal of the group algebra $\mathbb{F}_{q}[G]$ on which $G'$ acts trivially from the left (from the right). Finally, by Lemma~\ref{lem:Hnccspei} and hypothesis on the derived subgroup $G'$, we get the thesis.
\end{proof}

\subsection*{Acknowledgements} This note has been realized from the author’s Mater's degree thesis which was undertaken at Università degli Studi di Milano-Bicocca, and written under the patient supervision of Francesca Dalla Volta, to whom special thanks for her support and suggestions are due.

\end{document}